%% file: arxiv.tex
\documentclass[article]{llncs}
\usepackage{fullpage}[in]

\input{packages}

\title{Liquid democracy under vote correlation} 
\subtitle{On the fallacies of averaging and the excluded middle}

\titlerunning{Liquid Democracy under vote correlation}
\author{Seth Gilbert\inst{1}, Stefan Schmid\inst{2}, Santiago Schnell\inst{3}, Jakub Svoboda\inst{3}, and Michelle Yeo\inst{4,5}}

\institute{
  National University of Singapore, Singapore \and
  TU Berlin,  Germany \and
  Dartmouth College, USA \and
  Aarhus University, Denmark \and
  Nanyang Technological University, Singapore
}

\authorrunning{S. Gilbert, S. Schmid, S. Schnell, J. Svoboda, and M. Yeo}

\begin{document}

\maketitle

\begin{abstract}
Liquid democracy permits voters either to vote directly or to delegate their votes to other voters. Existing algorithmic analyses usually assign each voter a single scalar competence parameter, interpreted as an independent probability of voting for the ground truth. 
This representation is inadequate when delegation is fixed before later public information---for example, an audit, announcement, or misleading report---changes different voters' reliability in different ways.

In this work, we study a minimal common-signal model of this realistic phenomenon. Delegation occurs before a public binary signal is realised, while voting occurs afterwards. Conditional on the signal, sink votes are independent and each voter has a signal-specific competence; marginally, correctness events are correlated through the common signal. We show that delegation based on average competence is not a safe scalarisation: average-based delegation can violate do-no-harm, and a beneficial delegation rule may send votes to voters with lower average competence.

We then present and analyse three novel delegation mechanisms for this setting. 
The first is a conservative intersection mechanism that delegates only to neighbours whose competence exceeds the delegator's by a prescribed margin in every signal state; we show that the conservative intersection mechanism inherits all the guarantees of delegation in the scalar competence setting. 
The next mechanism is a confounded-set mechanism that enlarges the delegation pool to neighbours who are favoured in one state and worse by at most a prescribed tolerance in the other. For this mechanism we prove immediate and terminal expected-margin bounds and potential-function acyclicity, and explicity identify the weight-dispersion and mechanism-concentration conditions needed to obtain majority-correctness guarantees. Finally, on bounded-in-degree graphs, a multi-round certified-path mechanism propagates nonnegative two-dimensional path certificates; it is acyclic, yields statewise terminal competence improvement, and gives uniform bounds on path length and terminal voting weight.
\end{abstract}

\section{Introduction}
Liquid democracy is a hybrid voting paradigm in which voters may either vote directly or delegate their vote to another voter, with delegation being transitive.
Recent years have seen a surge in the use and study of liquid democracy, notably in blockchain decentralised autonomous organisation (DAO) voting~\cite{HallM24,FeichtingerFVW23,FritschMW24,SchmidS24} and political parties~\cite{KlingKHSS15,paulin20}.
In epistemic binary voting problems, where one of two alternatives is objectively correct, the normative appeal of delegation is straightforward: a less informed voter should be able to transfer voting weight to a better informed voter, thereby increasing the probability that the collective decision is correct.
The mathematical difficulty is that delegation changes not only the expected number of correct votes but also the \emph{distribution} of voting weights: concentrating many votes on a small number of sinks can increase apparent competence while reducing the reliability of the weighted majority~\cite{ChatterjeeG0SY25variance,KahngMP21}.

Most existing analyses represent a voter's competence by a single scalar, namely the probability that the voter selects the correct alternative. 
This scalar model is natural when correctness events are independent, or when all relevant uncertainty can be compressed into one parameter. In many decision environments, however, delegation is committed before later public information becomes available. A DAO may open a delegation window before an audit or security report is released; a political or organisational vote may occur after a public announcement, a change in issue framing, or a widely circulated item of misinformation. 
Such common events can affect groups differently, making one group more reliable and another less reliable. Political-science evidence also indicates that uncertainty and political context can alter voter preferences and indecision~\cite{Christensen_2022}. 
In these settings a voter is more naturally represented by a vector of signal-conditional competences than by a scalar average.

We study the simplest tractable instance of this phenomenon. 
Delegation takes place among voters before a binary public signal is realised, and the delegation induces a delegation graph where each outgoing edge between two voters represents vote delegation. Voting takes place after the signal among the sinks of the delegation graph. 
Conditional on the realised signal, sink votes are independent, but each voter has a signal-specific probability of voting correctly. We stress that the binary signal is not intended to model every source of correlation. 
Rather, it is a minimal common-cause model that preserves conditional independence while exposing the algorithmic consequences of state-dependent competence.

The main question we seek to address in our work is the following:

\smallskip
\emph{Given the positive results for delegation in the limited scalar-competence model, can we develop efficient delegation strategies in this new common-signal setting that achieve the same delegation desiderata as in the scalar-competence setting?}
\smallskip

We note that the extension from the scalar-competence model to our common-signal model is neither trivial nor automatic. 
In particular, four new difficulties arise. First, competence vectors are only partially ordered, and averaging them can lead to harmful delegations. 
Second, delegation may improve competence in one signal state while reducing it in the other, so scalar monotonicity and delegation acyclicity no longer follow. Third, an immediate expected improvement must be propagated through a random transitive delegation chain to the terminal sink. Fourth, as in the scalar model, majority correctness requires control of terminal voting weights and not merely an increase in mean competence. These are the technical issues addressed by the mechanisms and proofs below.

\subsection{Our contribution}
\smallskip{\noindent{\bf Common-signal model.}}
We introduce a common-signal model for correlated correctness in liquid democracy and define ex ante and statewise (that is, conditional on the signal state) versions of majority correctness, do-no-harm, and positive gain. 
Majority correctness decides the voting outcome, while do-no-harm and positive gain are classic desiderata that delegation mechanisms ought to satisfy~\cite{KahngMP21,HalpernHJMPR23,ChatterjeeG0SY25variance}.
Here, we emphasise that the statewise formulation is essential in our setting: a mechanism may be harmless after averaging over the signal prior while still being harmful conditional on one signal state.

\smallskip{\noindent{\bf The fallacy of averaging.}}
Our first technical result is negative. Delegation according to prior-weighted average competence is not a valid reduction from our signal-dependent model to the scalar model. 
We give two complete-graph constructions to highlight our negative result: average-based delegation can violate do-no-harm, and a beneficial delegation rule may send votes from higher-average voters to lower-average voters. The relevant object is therefore the geometry of the signal-conditional competence vector, not its scalar average.

\smallskip{\noindent{\bf Statewise dominance and subset-robust scalar reductions.}}
For each signal state $b$, a neighbour is \emph{favoured} when its competence exceeds the delegator's by at least a threshold $\alpha_b>0$. 
It is known in the scalar-competence setting that delegation mechanisms which satisfy the delegation desiderata invoke a process whereby voters check if a neighbour is favoured and then delegate to the favoured neighbour~\cite{ChatterjeeG0SY25variance}.
Our first mechanism replicates the spirit of the delegation mechanisms in the scalar-competence setting and delegates only to neighbours who are favoured in \emph{every signal state}. Conditional on state $b$, every selected edge is therefore a valid scalar approved edge. We formulate precisely the subset-robustness condition needed to transfer a scalar theorem, because the intersection mechanism samples from a subset of the full scalar approved set. The mechanism is most useful when relative expertise is stable across signal states---for example, when a domain expert remains reliably more accurate regardless of whether a later public report is accurate or misleading.

\smallskip{\noindent{\bf Confounded sets and the neutral middle.}}
An issue with the conservative intersection mechanism is that the pool of potential delegates is exactly voters which are favoured by others in both states.
This delegation pool size, however, can be insufficient to ensure we achieve the delegation desiderata in some cases.
We thus present our second delegation mechanism which aims to widen the pool of delegation candidates, while specifying additional conditions to achieve the delegation desiderata as in the scalar-competence setting.
For each signal state $b$, we additionally call a neighbour \emph{confounded} when its competence is worse than the delegator by more than a tolerance $\beta_b\ge0$; neighbours that are neither favoured nor confounded are \emph{neutral}. The phrase ``excluded middle'' is descriptive rather than a reference to the logical law of excluded middle: the error is to exclude the neutral middle region by treating ``not favoured'' as synonymous with ``confounded''. 
Our second mechanism uses neighbours who are favoured by at least $\alpha_b$ in one state and not confounded (i.e., not worse by at most $\beta_{1-b}$) in the other. We prove immediate expected-margin bounds, terminal expected-margin bounds by induction along a potential order, and potential-function acyclicity. We then state explicitly the additional weight-dispersion and mechanism-randomisation concentration hypotheses needed to convert these structural improvements into majority-correctness guarantees.

\smallskip{\noindent{\bf Certified paths with communication.}}
Our third mechanism is a bounded-radius distributed mechanism for bounded-in-degree graphs. It propagates nonnegative two-dimensional path certificates through favoured and neutral neighbours, thereby further widening the delegation pool by allowing delegation to fully neutral neighbours (disallowed in the confounded set and intersection mechanisms). We prove that every realised delegation graph is acyclic, every delegating voter terminates at a sink that is more competent in both signal states, and bounded in-degree yields uniform bounds on path length and terminal voting weight. The price of widening the delegation pool is additional communication rounds and a conservative coordinatewise comparison that may leave incomparable certificates unused.

\begin{table}[htb!]
\scriptsize
    \begin{tabular}{|c|c|c|}
    \hline
     Mechanism & Information and candidate set & Guarantee and trade-off  \\ \hline \hline
    Intersection (Alg.~\ref{alg:base})  & One round; $J^0(i)\cap J^1(i)$ & \makecell{Inherits subset-robust scalar guarantees statewise. Best for stable \\ statewise rankings;  may delegate little when rankings reverse.} 
    \\ \hline 
    Confounded set (Alg.~\ref{alg:two_dimensions}) & \makecell{One round; \\$J^b(i)\setminus K^{1-b}(i)$, mixed across states} & \makecell{Positive immediate and terminal  expected margins and acyclicity.\\ Uses a larger pool, but majority correctness  needs separate \\ weight and concentration control.}  \\ \hline
    Certified path (Alg.~\ref{alg:more_communication})  & \makecell{$T$ rounds; \\paths of admissible favoured/neutral edges } & \makecell{Acyclicity, statewise terminal improvement, and bounded  sink \\ weights on bounded-in-degree graphs. Costs communication \\ and may discard incomparable certificates.} \\ \hline
     \hline
    \end{tabular}
    \label{tab:mechanism-comparison}
    \caption{Qualitative comparison of the three mechanisms. ``One round'' means that a voter uses only labels on its immediate out-neighbours.}
\end{table}

\subsection{Related work}
\smallskip{\bf Algorithmic liquid democracy.}
Liquid democracy and delegated voting are well-studied~\cite{Alger2006,AmanatidisFLMP24incomplete,BloembergenGL19,ChatterjeeGSSY26opposing,Miller1969APF,Tullock1992}. On the algorithmic front, a central question is when delegated voting improves on direct voting in the epistemic binary-decision model~\cite{ButterworthB23,CaragiannisM19,CohensiusMMMO17,GreenArmytage2015}.
Kahng et al.~\cite{KahngMP21} show an impossibility for local delegation mechanisms on unrestricted graph topologies: no such mechanism can always be at least as accurate as direct voting while sometimes being strictly better. Halpern et al.~\cite{HalpernHJMPR23} obtain positive results on complete graphs. The closest predecessor is~\cite{ChatterjeeG0SY25variance}, which studies when local delegation is possible on natural graph classes and emphasises that delegation changes both mean competence and the variance of the weighted vote. The present work changes the competence model itself and addresses the vector-ordering, acyclicity, terminal-propagation, and concentration difficulties described above; it is not obtained by simply running the scalar proof twice.

\smallskip{\noindent\bf Correlated epistemic voting.}
The common-signal model is a simple common-cause model of correlated voter correctness and is related to Condorcet-type jury theorems with correlated voters. Pivato~\cite{pivato17} develops a general theory of epistemic democracy with correlated voters and shows that asymptotic correctness can persist when average covariance becomes sufficiently small. Our model is more specialised: conditional on a finite public signal, correctness events are independent, but unconditionally they are correlated. This specialisation permits the design of local mechanisms for signal-dependent competence vectors. Other work generalises Condorcet's jury theorem under specific dependence models~\cite{boland89,dietrichK13,ladha1992,LADHA1995}; Berg~\cite{berg93} compares positive and negative juror correlation.

\smallskip{\noindent\bf Information cascades and herding.}
The timing of a common signal after delegation and before voting is structurally related to information-cascade and herding models~\cite{banerjee92herding,BikhchandiHW92cascades,welch92cascades}. In those models, private signals are aggregated through sequential observable actions. Here the signal is common and external, delegation is fixed before it arrives, and only the later vote is signal-dependent.

\smallskip{\noindent\bf Computational social choice.}
Related extensions of liquid democracy include ranked delegations~\cite{brill22ranked}, ordinal elections~\cite{brillT18}, and issue-specific delegates~\cite{ChristoffG2017BinaryVW}. Caragiannis and Micha~\cite{CaragiannisM19} show in the scalar binary model that delegating to someone better informed on average need not be optimal and that approximately optimal delegation is generally NP-hard. Our lower-average construction exhibits an analogous ordering failure caused specifically by common-signal correlation.

\smallskip{\noindent\bf Epistemic democracy.}
The algorithmic study of whether delegated voting improves collective decisions is grounded in the philosophical literature on democracy as truth-discovery~\cite{estlund08,GoodinS18}.

\section{Model}\label{sec:model}

\subsection{Correlated delegated voting}\label{sec:model_correlations}

\smallskip{\bf Voters and awareness graph.}
We consider $n$ voters, $V = \{v_1, \dots, v_n\}$, that want to decide on some binary issue, represented by options $\{0,1\}$, by voting. 
We assume the existence of a ground truth with regard to the binary question that is unknown to the voters, which we denote by voting option $1$.
Voters are connected by a directed awareness graph $(V,E)$ where an edge $(v_i,v_j) \in E$ represents that voter $v_i$ is aware of another voter $v_j$. 
The \emph{neighbourhood} of a voter $v_i$, denoted by $Ne(v_i)$, is the set of voters that are directly connected to $v_i$, i.e., $Ne(v_i) := \{v_j \mid (v_i, v_j) \in E\}$.

\smallskip{\noindent{\bf Competences and external signal.}}
Although the voters do not know the ground truth outcome, each voter has some probability of voting for the correct option (i.e., option $1$), which we call the \emph{competence} of the voter.
In our setting, we assume that the competences of voters are correlated. 
Specifically, we assume a simple model of competence correlation based on an external, public, binary signal $s \in \{0,1\}$ that arrives \emph{in between} the vote delegation and actual voting processes (see the timeline depicted in~\Cref{fig:timeline}).
We use $\pi_b=\mathbb{P}[s=b]$ for $b\in\{0,1\}$, with $\pi_0+\pi_1=1$, to denote the signal prior.
The signal can represent a later public report, announcement, framing, or item of misinformation whose effect differs across voters. We represent its potential effect by two vectors $\mathbf{p}^0 = [p^0_1, p^0_2, \dots, p^0_n]$ and $ \mathbf{p}^1 = [p^1_1, p^1_2, \dots, p^1_n]$, where the $i$th element of $\mathbf{p}^0$ (resp. $\mathbf{p}^1$) represents the independent competence of voter $v_i$ when the value of the signal is $0$ (resp. $1$).

\begin{definition}[Problem instance]
    A problem instance is a tuple $G=(V, E,\mathbf{p}^0, \mathbf{p}^1, \pi)$, where $(V,E)$ is the directed awareness graph, $\mathbf{p}^0, \mathbf{p}^1$ are signal-conditional competence vectors, and $\pi = (\pi_0,\pi_1)$ is the signal prior.
\end{definition}

\begin{definition}[Average competence]\label{def:average}
The (prior-weighted) average competence of a voter \(v_i\) is
$$\bar p_i=\pi_0p_i^0+\pi_1p_i^1.$$

In the equiprobable case this reduces to \(\bar p_i=(p_i^0+p_i^1)/2\).
\end{definition}

\begin{figure}[htb!]
    \centering    
    \includegraphics[width=0.8\linewidth]{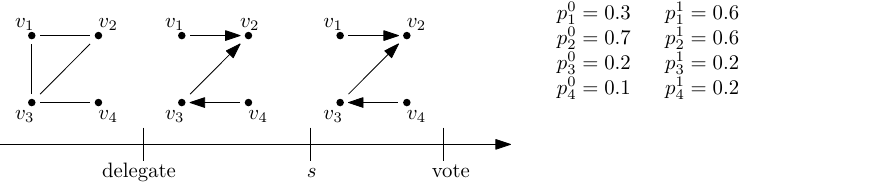}
    \caption{Timeline and information structure for the common-signal model. Delegation is fixed using the awareness graph before the public signal $s$ is realised; the signal then selects the relevant statewise competence vector, and the sinks of the delegation graph cast the weighted vote.}
    \label{fig:timeline}
\end{figure}

\smallskip{\noindent{\bf Vote delegation.}}
Each voter can delegate its vote to any single voter it is directly connected to in the awareness graph $G$, and this other voter can then transitively delegate all of its votes to someone else or vote directly. We formalise this process with the definition of a delegation mechanism as follows.

\begin{definition}[Delegation mechanism]
    A delegation mechanism $M$ maps a problem instance \(G\), together with any internal randomization of the mechanism, to a directed delegation graph \(G'=M(G)\) on vertex set \(V\). We restrict attention to mechanisms for which \(G'\) has out-degree at most one at every vertex and is acyclic. A vertex with out-degree zero in \(G'\) is called a sink. Let \(T(G')\) denote the set of sinks.

    For each sink \(v_i\in T(G')\), let \(w_i(G')\) be the number of voters whose delegation chain terminates at \(v_i\), including \(v_i\) itself. Thus \(\sum_{v_i\in T(G')}w_i(G')=n\).
\end{definition}

The final decision is made based on weighted majority vote among the sinks in the delegation graph $G'$, breaking ties in favour of the ground truth voting outcome.

\begin{definition}[Weighted majority correctness]
Fix a delegation graph \(G'\) and a signal value \(b\). For each sink \(v_i\in T(G')\), let

  $$X_i^b\sim \operatorname{Bernoulli}(p_i^b)$$

be the indicator that the sink \(v_i\) votes for the correct alternative, with the variables \(\{X_i^b:v_i\in T(G')\}\) independent conditional on \(s=b\). The weighted majority is correct if
$$\sum_{v_i\in T(G')} w_i(G')X_i^b \ge \frac{n}{2}.$$

\end{definition}


\smallskip{\noindent{\bf Local mechanisms and label oracles.}}
We distinguish the \emph{epistemic instance} used to analyse correctness from the \emph{local information} available to our algorithms. The vectors $\mathbf p^0,\mathbf p^1$ are latent parameters of the epistemic model; a voter does not observe their numerical values, the ground truth, or the total number of voters.

Instead, each voter $v_i$ receives a finite label for each out-neighbour from an exogenous local label oracle,
$$\mathcal I_i(G)=\bigl(Ne(v_i),\ell_i,\prec_i\bigr),$$
where $\ell_i(v_j)$ is a finite-valued label and $\prec_i$ is a fixed local tie-breaking order or source of local randomness. Locality refers to the computation performed \emph{after} these labels are supplied: the decision of $v_i$ may depend only on labels attached to vertices in $Ne(v_i)$ and on local randomness. It does not assert that the labels themselves can be statistically learned from one-hop information.

Looking ahead, in \Cref{sec:mechanisms} the label $\ell_i(v_j)$ consists of four comparison bits indicating whether $v_j$ is favoured or confounded relative to $v_i$ in each signal state. These bits are generated from the latent competence vectors for the purpose of analysis, but the algorithms never query the numerical values $p_i^b$. This is the same modelling separation as the approved-neighbour information primitive in the scalar setting~\cite{ChatterjeeG0SY25variance}. In an implementation, labels could be supplied by a calibration, reputation, or certification layer; learning reliable labels is outside the scope of this paper and is discussed as future work.

\begin{example}[Delegation based on favoured set size]\label{ex:setsize}
Suppose the local oracle gives voter $v_i$ a set $J(i)\subseteq Ne(v_i)$ of favoured neighbours. A local mechanism may check whether $|J(i)|$ exceeds a function of the local degree and, if so, delegate to a uniformly random member of $J(i)$; otherwise the voter votes directly. The scalar mechanism studied in the independent-competence setting is the special case
$$J(i)=\{v_j\in Ne(v_i):p_j-p_i\ge\alpha\}.$$
\end{example}

\begin{example}[Delegation based on average competence]\label{ex:average}
The average-competence label oracle declares $v_j\in Ne(v_i)$ favoured whenever $\bar p_j>\bar p_i$. The corresponding mechanism delegates uniformly to such a neighbour when the set is nonempty and otherwise votes directly. This example is used only as a negative benchmark in \Cref{sec:strawman}.
\end{example}

\begin{example}[Direct voting]
Let $D$ be the mechanism that never delegates. Every voter is a sink of weight one. 
\end{example}

\smallskip\noindent{\bf Probability of correct outcome.}
We now define the probability of achieving the correct outcome given an input problem instance $G$ and a delegation mechanism $M$. Due to the signal prior in our setting, we distinguish between \emph{statewise} and \emph{ex ante} success probabilities: the former measures the success probability of $M$ given a specific signal state $b$, and the latter measures the expected success probability of $M$ over all signal states.

\begin{definition}[Statewise and ex ante success probabilities]
For a mechanism \(M\), define the statewise success probability
$$ P_M^b(G)=\mathbb{P}\left(
  \sum_{v_i\in T(M(G))}w_i(M(G))X_i^b\ge \frac{n}{2}
  \;\middle|\; s=b
  \right),$$
where the probability is over the mechanism's internal randomisation and over the conditional Bernoulli votes of the sinks. The ex ante success probability is
$$P_M(G)=\sum_{b\in\{0,1\}}\pi_b P_M^b(G).$$

Let \(D\) denote direct voting, i.e., the mechanism for which every voter is a sink with weight one. Define
$$\gain_b(M,G)=P_M^b(G)-P_D^b(G)$$
and
$$\gain(M,G)=P_M(G)-P_D(G)=\sum_{b\in\{0,1\}}\pi_b\gain_b(M,G).$$
\end{definition}

We use the term \emph{loss} to denote negative gain.





\subsection{Delegation desiderata}\label{sec:desiderata}
We define some important desiderata that our delegation mechanisms need to satisfy. 

\smallskip\noindent{\bf  Do-no-harm (DNH).}
The first desideratum we want our mechanisms to satisfy is do-no-harm. Informally, this property 
ensures that as our restricted problem instances grow in size, the loss of our mechanisms compared to direct voting goes to 0. We extend the notion of DNH in the scalar competence model to account for the signal prior in our common signal model by defining two notions of DNH: ex ante and statewise. 

\begin{definition}[Ex ante DNH]
Let \(\mathcal{G}_n\) be a class of problem instances with \(n\) voters. A mechanism \(M\) satisfies ex ante DNH on \((\mathcal{G}_n)_{n\ge1}\) if
$$\liminf_{n\to\infty}\ \inf_{G\in\mathcal{G}_n}\ \gain(M,G)\ge 0.$$
  
Equivalently, for every \(\varepsilon>0\), there exists \(n_0\) such that for all \(n\ge n_0\) and all \(G\in\mathcal G_n\),
$$\gain(M,G)>-\varepsilon.$$
  
\end{definition}

\begin{remark}
The infimum is intentional. Even for fixed $n$, the competence vectors vary continuously, so $\mathcal G_n$ is generally infinite and the worst-case gain need not be attained. A minimum may replace the infimum only under additional compactness and attainment assumptions.
\end{remark}

\begin{definition}[Statewise DNH]
A mechanism \(M\) satisfies statewise DNH on \((\mathcal{G}_n)_{n\ge1}\) if, for each \(b\in\{0,1\}\),
$$\liminf_{n\to\infty}\ \inf_{G\in\mathcal{G}_n}\ \gain_b(M,G)\ge 0.$$
\end{definition}
Note that statewise DNH implies ex ante DNH for every signal prior \(\pi\).


\smallskip\noindent{\bf Positive gain (PG)}
The second desideratum that we want our mechanisms to achieve is positive gain. Positive gain ensures that as problem instances grow in size there are some instances where delegated mechanisms perform better (i.e., achieve positive gain) compared to direct voting. 

\begin{definition}[Positive gain]\label{def:pg}
A mechanism \(M\) satisfies positive gain on \((\mathcal{G}_n)_{n\ge1}\) if there exist constants \(\gamma>0\) and \(n_0\) such that for every \(n\ge n_0\), there exists an instance \(G\in\mathcal{G}_n\) satisfying
$$\gain(M,G)\ge \gamma.$$
\end{definition}

A statewise version of positive gain is obtained by replacing \(\gain\) with \(\gain_b\) for a specified signal value \(b\), or by requiring the inequality for every \(b\).

\smallskip\noindent{\bf Strong positive gain (SPG)}
The positive-gain definition in~\Cref{def:pg} is existential. To record how much a mechanism helps when it delegates at a non-negligible scale, we use the following delegation-rate profile. 
We stress that this is deliberately a best-case profile; a worst-case high-delegation guarantee would require replacing the supremum below by an infimum and is not claimed here.

\begin{definition}[Delegation-rate gain profile]\label{def:gain-profile}
For $\rho\in(0,1)$, define
$$\Gamma_n(\rho)=
\sup\Bigl\{\gain(M,G):G\in\mathcal G_n,\;
\mathbb E[\#\{\text{voters who delegate under }M\}]\ge \rho n\Bigr\}.$$
The mechanism has \emph{existential strong positive gain} at delegation rate $\rho$ if
$$\liminf_{n\to\infty}\Gamma_n(\rho)>0.$$
\end{definition}

\begin{remark}
The restriction $\rho<1$ is necessary for finite acyclic delegation graphs: at least one vertex must be a sink, so not all voters can delegate. In the following sections, when we say that a mechanism has strong positive gain, we mean the existential version in \Cref{def:gain-profile}. Any statement asserting a uniform guarantee over all high-delegation instances will be stated separately.
\end{remark}

\section{Local delegation mechanisms}\label{sec:mechanisms}

\subsection{Why average competence is not a safe reduction}\label{sec:strawman}
Let $M$ denote the average-competence delegation mechanism from~\Cref{ex:average}: voter $v_i$ delegates to a neighbouring voter $v_j$ only when $\bar p_j>\bar p_i$. The next two propositions show that this scalarisation is not a valid reduction from the common-signal model to the scalar-competence model. Exact finite-size illustrations appear in~\Cref{fig:sanity-average-failure,fig:sanity-lower-average}.

\begin{proposition}[Average-based delegation can violate DNH]\label{prop:average-violates-dnh}
Assume $\pi_0=\pi_1=1/2$. For every sufficiently large $n$, there is a complete-graph instance with $n$ voters for which direct voting succeeds with probability tending to $1$ as $n\to\infty$, while the mechanism that delegates to a neighbour of strictly larger average competence succeeds with probability $0.7$. Consequently, average-based delegation violates ex ante do-no-harm on complete graphs.
\end{proposition}

\begin{proof}
Let $v^\ast$ be one distinguished voter with $(p_{v^\ast}^0,p_{v^\ast}^1)=(1,0.4)$, and let each of the remaining $n-1$ voters have competence vector $(0.6,0.6)$. The graph is complete. Under direct voting, conditional on $s=0$, all voters have competence at least $0.6$, and $v^\ast$ is always correct. Conditional on $s=1$, the $n-1$ ordinary voters have competence $0.6$, while $v^\ast$ has competence $0.4$. In either signal state, the expected number of correct direct votes exceeds $n/2$ by a quantity of order $n$. Hoeffding's inequality therefore implies
$$P_D^0(G_n)\to1,\qquad P_D^1(G_n)\to1,$$
and hence $P_D(G_n)\to1$.

The average competence of $v^\ast$ is $(1+0.4)/2=0.7$, whereas that of every ordinary voter is $0.6$. Thus all ordinary voters delegate to $v^\ast$, while $v^\ast$ votes directly. The final decision is exactly $v^\ast$'s vote, and
$$P_M(G_n)=\frac12\cdot1+\frac12\cdot0.4=0.7.$$
Therefore $\gain(M,G_n)\to-0.3$. \qed
\end{proof}

\begin{remark}
For this numerical instance, an adversary that always selects the adverse state obtains limiting delegated success probability $0.4$ and limiting loss $0.6$. The loss can be made arbitrarily close to $1$ by replacing $(1,0.4)$ with $(1,\eta)$ and the ordinary voters' competence with $1/2+\eta/4$, then letting $\eta\downarrow0$.
\end{remark}

\begin{proposition}[Beneficial delegation may go to lower-average voters]\label{prop:lower-average-beneficial}
Assume $\pi_0=\pi_1=1/2$. For every sufficiently large $m$, there is a complete-graph instance with $2m$ voters in which direct voting succeeds with limiting probability $1/2$, while a delegation rule that sends voters to strictly lower-average voters succeeds with limiting probability $1$.
\end{proposition}

\begin{proof}
Partition the voters into groups $A$ and $B$, each of size $m$. Voters in $A$ have competence vector $(0.6,0.6)$, while voters in $B$ have competence vector $(0.3,1)$. Their average competences are $0.6$ and $0.65$, respectively, so delegation from $B$ to $A$ is delegation to lower-average voters.

Under direct voting, if $s=1$, the expected number of correct votes exceeds the threshold $m$ by order $m$, so $P_D^1(G_m)\to1$. If $s=0$, the expected number is $0.6m+0.3m=0.9m<m$, so $P_D^0(G_m)\to0$. Hence $P_D(G_m)\to1/2$.

Pair each voter in $B$ with a distinct voter in $A$ and have each $B$-voter delegate to its paired $A$-voter. The $A$-voters are then independent sinks of weight $2$ and competence $0.6$ in both states. Their weighted majority is correct exactly when a majority of the $A$-sinks is correct, which occurs with probability tending to $1$ by Hoeffding's inequality. \qed
\end{proof}

\begin{remark}
This proposition is an existence result showing that average competence gives the wrong ordering. It is not a claim that Algorithm~\ref{alg:two_dimensions} necessarily selects this paired delegation pattern.
\end{remark}

The following two plots are exact finite-size calculations for the constructions in Propositions~\ref{prop:average-violates-dnh} and~\ref{prop:lower-average-beneficial} in~\Cref{sec:strawman}; no Monte Carlo sampling is used. They are included as a diagnostic illustration of the asymptotic arguments. 

\begin{figure}[htb!]
    \centering
    \includegraphics[width=.52\linewidth]{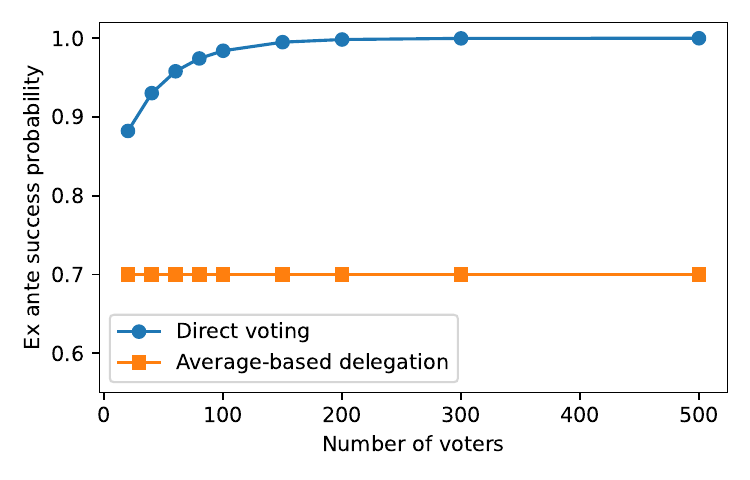}
    \caption{Exact finite-$n$ success probabilities for Proposition~\ref{prop:average-violates-dnh}. Direct voting tends rapidly to success probability one, while average-based delegation is fixed at $0.7$ because all ordinary voters delegate to the volatile voter $v^\ast$.}
    \label{fig:sanity-average-failure}
\end{figure}

\begin{figure}[htb!]
    \centering
    \includegraphics[width=.52\linewidth]{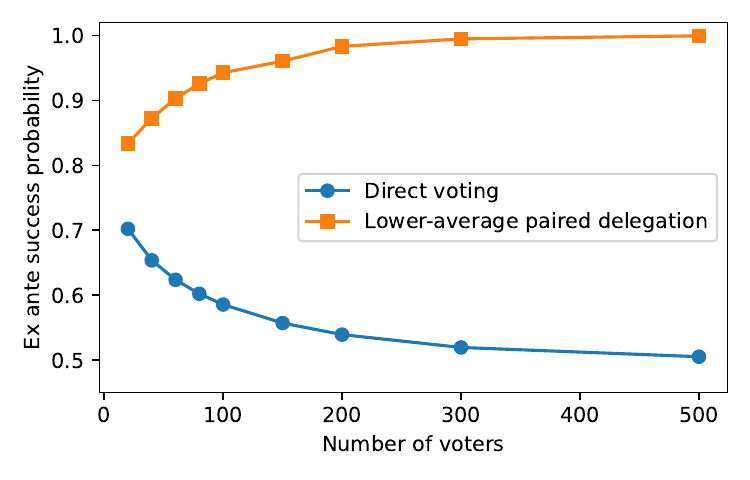}
    \caption{Exact finite-size success probabilities for Proposition~\ref{prop:lower-average-beneficial}. The lower-average paired-delegation rule tends to success probability one, whereas direct voting tends to $1/2$. As noted in the main text, this is an existence result about the failure of average ordering, not a claim that Algorithm~\ref{alg:two_dimensions} necessarily selects this delegation pattern.}
    \label{fig:sanity-lower-average}
\end{figure}

\subsection{Conservative intersection mechanism}
For $b\in\{0,1\}$ and a threshold $\alpha_b>0$, define the statewise favoured set
$$J^b(i)=\{v_j\in Ne(v_i):p_j^b-p_i^b\ge\alpha_b\}.$$
Thus ``favoured'' means better by at least the explicit margin $\alpha_b$ in signal state $b$. The conservative candidate set is
$$J^\cap(i)=J^0(i)\cap J^1(i).$$
In other words, a voter in $J^\cap(i)$ is better than $v_i$ by at least $\alpha_b$ in every state.

\begin{algorithm}[htb!]
  \begin{algorithmic}[1]
    \Require threshold function $j$
    \Ensure decisions to vote or delegate to another voter
    \For{each voter $v_i$}
    \State $d_i\gets |Ne(v_i)|$
    \State compute $J^\cap(i)=J^0(i)\cap J^1(i)$
    \If{$|J^\cap(i)|\ge j(d_i)$}
        \State $v_f\gets\Call{RandomChoice}{J^\cap(i)}$
        \State $\Call{Delegate}{v_f}$
    \Else
        \State $\Call{Vote}{v_i}$
    \EndIf
    \EndFor
  \end{algorithmic}
  \caption{Conservative intersection delegation}\label{alg:base}
\end{algorithm}

\begin{definition}[Subset-robust scalar guarantee]\label{def:subset-robust}
Fix a scalar competence vector $\mathbf p$, an approval margin $\alpha>0$, and a local threshold function $j$. A scalar guarantee is \emph{subset-robust} if it remains valid when, for every voter $v_i$, the mechanism may sample from any locally supplied set
$$S_i\subseteq \{v_j\in Ne(v_i):p_j-p_i\ge\alpha\}$$
with $|S_i|\ge j(|Ne(v_i)|)$ whenever it delegates, rather than being required to sample from the full approved set.
\end{definition}

\begin{restatable}[Statewise reduction for the intersection mechanism]{theorem}{intersection}\label{thm:intersection-reduction}
Fix a graph family and scalar parameter regime. Suppose a scalar DNH, PG, or SPG theorem with approval margin $\alpha_b$ and threshold function $j$ is subset-robust in the sense of~\Cref{def:subset-robust}. Then Algorithm~\ref{alg:base} inherits the same conclusion conditional on each signal state $b$, and therefore inherits the corresponding ex ante conclusion. For PG or SPG, the scalar witness can be embedded by setting $\mathbf p^0=\mathbf p^1$ and $\alpha_0=\alpha_1$.
\end{restatable}


\begin{proof}

Fix a signal value $b\in\{0,1\}$. Conditional on $s=b$, the correctness indicators are independent Bernoulli variables with competence vector $p^b$. Moreover, every edge selected by the conservative intersection mechanism from $v_i$ to $v_j$ satisfies $p_j^b-p_i^b\ge\alpha_b$, because $v_j\in J^\cap(i)\subseteq J^b(i)$.

The only point requiring care is the sampling rule. Conditional on state $b$, the mechanism samples from $J^\cap(i)$, which is generally a proper subset of $J^b(i)$. Hence the scalar theorem applies directly only if it is subset-robust in the sense stated in \Cref{thm:intersection-reduction}: any sufficiently large subset of approved neighbours may be used. Under this hypothesis, the scalar do-no-harm theorem gives
$$\liminf_{n\to\infty}\inf_{G\in\mathcal G_n}\gain_b(M,G)\ge0.$$

Since this holds for both signal states, statewise do-no-harm follows. Ex ante do-no-harm follows by averaging:
$$\gain(M,G)=\pi_0\gain_0(M,G)+\pi_1\gain_1(M,G).$$

For positive gain, embed the scalar model by setting $p^0=p^1=p$ and $\alpha_0=\alpha_1=\alpha$. Then $J^\cap(i)=J(i)$, so the conservative mechanism coincides with the scalar mechanism on the embedded instance. \qed

\end{proof}

\smallskip\noindent{\bf When is the intersection nonempty?}
The mechanism is enabled when relative expertise is stable across signal states. For example, if a technical expert remains at least $\alpha_b$ more reliable than a neighbour whether a later report is accurate or misleading, that expert lies in the intersection. By contrast, when statewise rankings reverse, $J^0(i)\cap J^1(i)$ may be empty even though useful trade-offs exist. A neighbour may be much better in state $0$ and only slightly worse in state $1$; the next mechanism is designed to use such neighbours.

\subsection{Confounded-set mechanism and the neutral middle}
For $b\in\{0,1\}$, let $\alpha_b>0$ be the favoured margin and $\beta_b\ge0$ the tolerated loss. Define
$$J^b(i)=\{v_j\in Ne(v_i):p_j^b-p_i^b\ge\alpha_b\},\qquad
K^b(i)=\{v_j\in Ne(v_i):p_i^b-p_j^b>\beta_b\}.$$
Thus $J^b(i)$ contains neighbours better by at least $\alpha_b$, whereas $K^b(i)$ contains neighbours worse by more than $\beta_b$. The remaining neighbours are neutral in state $b$. The phrase ``excluded middle'' refers to the erroneous exclusion of this neutral region: $v_j\notin J^b(i)$ does not imply $v_j\in K^b(i)$.

Because $J^b(i)\subseteq Ne(v_i)$, the cross-state acceptable sets simplify to
$$A^0(i)=J^0(i)\setminus K^1(i),\qquad A^1(i)=J^1(i)\setminus K^0(i).$$
A voter in $A^0(i)$ is better by at least $\alpha_0$ in state $0$ and worse by at most $\beta_1$ in state $1$; the definition of $A^1(i)$ is symmetric. Figure~\ref{fig:placeholder} depicts these regions.

\begin{figure}
    \centering
    \includegraphics[width=0.4\textwidth]{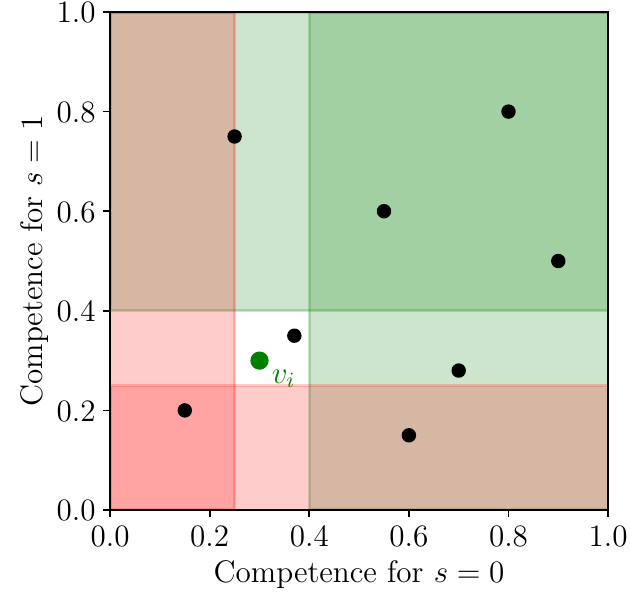}
    \caption{Geometry of favoured, confounded, and neutral sets for a fixed voter $v_i$. The horizontal coordinate is competence when $s=0$, and the vertical coordinate is competence when $s=1$. The margins $\alpha_b$ define the favoured regions, whereas $\beta_b$ define the confounded regions. The cross-state acceptable sets $A^0(i)$ and $A^1(i)$ admit neighbours that are strongly better in one state and not worse by more than the tolerated amount in the other.}
    \label{fig:placeholder}
\end{figure}

\begin{algorithm}[htb!]
  \begin{algorithmic}[1]
    \Require local labels $J^0(i),J^1(i),K^0(i),K^1(i)$; local threshold functions $c,c'$; mixing parameter $\lambda\in[0,1]$
    \Ensure decisions to vote or delegate to another voter
    \For{each voter $v_i$}
    \State $d_i\gets |Ne(v_i)|$
    \State $A^0(i)\gets J^0(i)\setminus K^1(i)$
    \State $A^1(i)\gets J^1(i)\setminus K^0(i)$
    \If{$|J^0(i)\cap J^1(i)|\ge c(d_i)$}
        \State $v_f\gets\Call{RandomChoice}{J^0(i)\cap J^1(i)}$
        \State $\Call{Delegate}{v_f}$
    \ElsIf{$|A^0(i)|\ge c'(d_i)$ and $|A^1(i)|\ge c'(d_i)$}
        \State draw $B_i\in\{0,1\}$ with $\mathbb P(B_i=0)=\lambda$
        \State $v_f\gets\Call{RandomChoice}{A^{B_i}(i)}$
        \State $\Call{Delegate}{v_f}$
    \Else
        \State $\Call{Vote}{v_i}$
    \EndIf
    \EndFor
  \end{algorithmic}
  \caption{Delegation with favoured and confounded sets}\label{alg:two_dimensions}
\end{algorithm}

\begin{lemma}[Potential-function acyclicity]\label{lem:potential-acyclicity}
Assume $\alpha_0\alpha_1>\beta_0\beta_1$. Then there exist constants $r_0,r_1>0$ such that
$$\Phi_i=r_0p_i^0+r_1p_i^1$$
strictly increases along every delegation edge selected by \Cref{alg:two_dimensions}. Consequently, every realised delegation graph is acyclic.
\end{lemma}

\begin{proof}
Choose $r_0,r_1>0$ such that
$$r_0\alpha_0>r_1\beta_1,
\qquad
r_1\alpha_1>r_0\beta_0.$$
Equivalently,
$$\frac{\beta_0}{\alpha_1}<\frac{r_1}{r_0}<\frac{\alpha_0}{\beta_1},$$
with the evident interpretation when a $\beta_b$ is zero. Such a ratio exists exactly when $\alpha_0\alpha_1>\beta_0\beta_1$.

If $v_j\in J^0(i)\cap J^1(i)$, both coordinates increase by at least $\alpha_0$ and $\alpha_1$, so $\Phi_j>\Phi_i$. If $v_j\in A^0(i)$, then $p_j^0-p_i^0\ge\alpha_0$ and $p_j^1-p_i^1\ge-\beta_1$, whence
$$\Phi_j-\Phi_i\ge r_0\alpha_0-r_1\beta_1>0.$$

The case $v_j\in A^1(i)$ is symmetric. A directed cycle would force a strict increase of $\Phi$ around a closed path, which is impossible. \qed
\end{proof}

\begin{proposition}[Immediate expected statewise margin]\label{prop:expected-margin-lambda}
Suppose \Cref{alg:two_dimensions} reaches its second branch for voter $v_i$, and let $v_f$ be the selected immediate delegate. Then
$$\mathbb E[p_f^0-p_i^0]\ge \lambda\alpha_0-(1-\lambda)\beta_0,
\qquad
\mathbb E[p_f^1-p_i^1]\ge (1-\lambda)\alpha_1-\lambda\beta_1.$$
The expected immediate margin is positive in both states whenever
$$\frac{\beta_0}{\alpha_0+\beta_0}<\lambda<\frac{\alpha_1}{\alpha_1+\beta_1},$$
and such a $\lambda$ exists if and only if $\alpha_0\alpha_1>\beta_0\beta_1$.
\end{proposition}

\begin{proof}
If $B_i=0$, then $v_f\in A^0(i)$. Hence $p_f^0-p_i^0\ge\alpha_0$, and $v_f\notin K^1(i)$ gives $p_f^1-p_i^1\ge-\beta_1$. If $B_i=1$, then $p_f^1-p_i^1\ge\alpha_1$ and $p_f^0-p_i^0\ge-\beta_0$. Averaging over $B_i$ gives the displayed bounds. The two bounds are positive exactly when $\lambda(\alpha_0+\beta_0)>\beta_0$ and $\lambda(\alpha_1+\beta_1)<\alpha_1$, which gives the interval. The interval is nonempty exactly when $\alpha_0\alpha_1>\beta_0\beta_1$.  \qed
\end{proof}

\begin{corollary}[A concrete bounded-weight case for \Cref{alg:two_dimensions}]\label{cor:alg2-bounded-weight}
Assume the hypotheses of Lemma~\ref{lem:potential-acyclicity}, and fix any admissible choice of $r_0,r_1$. Let
$$\eta=\min\{r_0\alpha_0+r_1\alpha_1,\; r_0\alpha_0-r_1\beta_1,\; r_1\alpha_1-r_0\beta_0\}>0.$$

If the awareness graph has maximum in-degree at most $\Delta_{\mathrm{in}}$, then every realised delegation path under \Cref{alg:two_dimensions} has length at most
$$L_\Phi=\left\lceil \frac{r_0+r_1}{\eta}\right\rceil,$$
and every terminal sink has weight at most
$$W_\Phi=1+\Delta_{\mathrm{in}}+\cdots+\Delta_{\mathrm{in}}^{L_\Phi}.$$
Consequently every realised delegation graph $H$ satisfies $W_2(H)\le W_\Phi n$.
\end{corollary}

\begin{proof}
Along every realised delegation edge, the potential $\Phi_i=r_0p_i^0+r_1p_i^1$ increases by at least $\eta$. Since $0\le p_i^0,p_i^1\le1$, the potential lies in $[0,r_0+r_1]$, so no directed delegation path can have more than $L_\Phi$ edges. With maximum in-degree $\Delta_{\mathrm{in}}$, at most $1+\Delta_{\mathrm{in}}+\cdots+\Delta_{\mathrm{in}}^{L_\Phi}$ vertices can reach any fixed terminal sink. The displayed $W_2$ bound follows from $w_i\le W_\Phi$ and $\sum_i w_i=n$. \qed
\end{proof}

\begin{lemma}[Terminal expected margin for \Cref{alg:two_dimensions}]\label{lem:terminal-expected-margin}
Under the hypotheses of \Cref{thm:disaproval}, let $\tau(i)$ be the terminal sink reached by voter $v_i$. Then for $b\in\{0,1\}$,
$$\mathbb E[p_{\tau(i)}^b-p_i^b]\ge \underline\delta_b\,\mathbb P(v_i\text{ delegates}),$$
with $\underline\delta_0$ and $\underline\delta_1$ as in \Cref{thm:disaproval}.
\end{lemma}

\begin{proof}
Order the voters by decreasing potential $\Phi$. By \Cref{lem:potential-acyclicity}, every possible delegation edge points to a strictly larger potential, so this order is topological for every realised delegation graph. For each state $b$, let
$$Q_i^b=\mathbb E[p_{\tau(i)}^b]$$

where the expectation is over the mechanism's randomisation downstream of $v_i$.

We prove by induction in decreasing potential that $Q_i^b\ge p_i^b$ for every voter, and that the displayed stronger bound holds when $v_i$ delegates. If $v_i$ votes directly, then $Q_i^b=p_i^b$. If $v_i$ delegates through the intersection branch, then each immediate delegate $v_f$ satisfies $p_f^b-p_i^b\ge\alpha_b$, and the induction hypothesis gives $Q_f^b\ge p_f^b$. Hence $Q_i^b\ge p_i^b+\alpha_b$. If $v_i$ delegates through the second branch, \Cref{prop:expected-margin-lambda} gives the corresponding immediate expected margin, and the induction hypothesis again gives $Q_f^b\ge p_f^b$. Thus $Q_i^b\ge p_i^b+\underline\delta_b$. Combining this with the direct-voting case proves the claim. \qed
\end{proof}

\begin{restatable}[Structural guarantees and conditional majority-correctness reduction for Algorithm~\ref{alg:two_dimensions}]{theorem}{disapproval}\label{thm:disaproval}
Assume $\alpha_0\alpha_1>\beta_0\beta_1$ and choose $\lambda$ such that
$$\frac{\beta_0}{\alpha_0+\beta_0}<\lambda<\frac{\alpha_1}{\alpha_1+\beta_1}.$$
Then every realised delegation graph is acyclic. If $\tau(i)$ is the terminal sink reached by voter $v_i$, then for each state $b$,
$$\mathbb E[p^b_{\tau(i)}-p_i^b]\ge \underline\delta_b\,\mathbb P(v_i\text{ delegates}),$$
where
$$\underline\delta_0=\lambda\alpha_0-(1-\lambda)\beta_0,\qquad
\underline\delta_1=(1-\lambda)\alpha_1-\lambda\beta_1.$$
Consequently,
$$\mathbb E[\mu_{M(G)}^b]\ge \mu_D^b+\underline\delta_b\,\mathbb E[d(M(G))],$$
where $d(M(G))$ is the number of delegating voters. If, in addition, the induced terminal weights and the realised terminal mean satisfy the same weight-dispersion and mechanism-randomisation concentration hypotheses used by a scalar DNH, PG, or SPG theorem, then Algorithm~\ref{alg:two_dimensions} inherits the corresponding statewise and ex ante conclusion.
\end{restatable}

\begin{proof}
Acyclicity is proved in~\Cref{lem:potential-acyclicity}; the immediate and terminal expected-margin statements are proved in~\Cref{prop:expected-margin-lambda,lem:terminal-expected-margin}. Summing the terminal inequality over voters gives the displayed mean bound, since
$$\mu_{M(G)}^b=\sum_{v_i\in T(M(G))}w_i(M(G))p_i^b=\sum_i p_{\tau(i)}^b.$$
Mean improvement alone does not determine majority correctness. Conditional on a realised delegation graph, weighted Hoeffding depends on $W_2(H)=\sum_iw_i^2$, while the mechanism's randomisation must also keep the realised terminal mean close to its expectation. Under these two additional hypotheses, the scalar concentration argument applies separately in each signal state; averaging over $b$ gives the ex ante conclusion.
\end{proof}

\begin{remark}
The terminal-weight condition is concrete rather than vacuous. If the awareness graph has maximum in-degree at most $\Delta_{\rm in}$, the potential in~\Cref{lem:potential-acyclicity} increases by at least a constant $\eta>0$ on each delegation edge and lies in a bounded interval. Hence every path has length at most $\lceil(r_0+r_1)/\eta\rceil$, and every sink has bounded weight; see~\Cref{cor:alg2-bounded-weight}. Dense graph classes require a separate dispersion argument.
\end{remark}

\begin{remark}[Scope of the conditional reduction]
For Algorithm~\ref{alg:two_dimensions}, the preceding lemmas establish mean improvement and acyclicity. Corollary~\ref{cor:alg2-bounded-weight} verifies the terminal-weight component on bounded-in-degree graphs. A full DNH or SPG conclusion in a dense or minimum-degree regime additionally requires the mechanism-randomisation concentration argument of the corresponding scalar theorem; we do not restate an incomplete adaptation here.

\end{remark}

\section{Trading locality for communication on bounded-degree graphs}\label{sec:distributed}
The one-round mechanisms in~\Cref{sec:mechanisms} use only labels on a voter's immediate out-neighbours. On bounded-degree graphs, additional synchronous communication rounds can reveal short safe paths through neutral intermediaries, further widening the delegate pool. The purpose is not to follow an arbitrary path: every accepted path must carry a nonnegative certificate in each signal state.

Call a directed awareness edge $(v_i,v_j)$ \emph{admissible} when its head is not confounded for $v_i$ in either state,
$$v_j\notin K^0(i)\cup K^1(i).$$
For an admissible edge, define the score vector
$$e_{ij}=(e_{ij}^0,e_{ij}^1),\qquad
e_{ij}^b=\begin{cases}
1,&v_j\in J^b(i),\\
-1,&v_j\notin J^b(i).
\end{cases}$$
Thus $e_{ij}^b=1$ certifies an actual competence gain of at least $\alpha_b$ in state $b$, whereas $e_{ij}^b=-1$ certifies that the edge is neutral rather than confounded and hence incurs a loss of at most $\beta_b$. The certificate of a directed path is the coordinatewise sum of its edge scores.

Write $x\succeq y$ for coordinatewise dominance and $x\succ y$ when $x\succeq y$ and $x\ne y$. Algorithm~\ref{alg:more_communication} propagates only certificates that are coordinatewise nonnegative. It deliberately avoids scalarising incomparable certificate vectors. If two candidates become incomparable after one has been selected, the fixed local scan order determines which certificate is retained. This conservative convention may miss feasible paths, but it makes certificate propagation monotone and supports the cycle argument below.

\begin{algorithm}[htb!]
  \begin{algorithmic}[1]
    \Require local labels $J^0(i),J^1(i),K^0(i),K^1(i)$; thresholds $\alpha_b>\beta_b$; round bound $T$; fixed local neighbour order
    \Ensure decision to vote or delegate to a neighbour
    \For{each voter $v_i$ in parallel}
        \State $R_i\gets(0,0)$; $D_i\gets\bot$
    \EndFor
    \For{$t=1,\ldots,T$}
        \For{each voter $v_i$ in parallel}
            \State $R_i^{\rm new}\gets R_i$; $D_i^{\rm new}\gets D_i$
            \For{each admissible edge $(v_i,v_j)$ in the fixed local order}
                \State $C\gets R_j+e_{ij}$
                \If{$C\succeq(0,0)$ and $C\succ R_i^{\rm new}$}
                    \State $R_i^{\rm new}\gets C$; $D_i^{\rm new}\gets v_j$
                \EndIf
            \EndFor
        \EndFor
        \For{each voter $v_i$ in parallel}
            \State $R_i\gets R_i^{\rm new}$; $D_i\gets D_i^{\rm new}$
        \EndFor
    \EndFor
    \For{each voter $v_i$ in parallel}
        \If{$D_i=\bot$}
            \State $\Call{Vote}{v_i}$
        \Else
            \State $\Call{Delegate}{D_i}$
        \EndIf
    \EndFor
  \end{algorithmic}
  \caption{Certified-path delegation with $T$ synchronous communication rounds}\label{alg:more_communication}
\end{algorithm}


\begin{restatable}[Certified-path guarantees on bounded-in-degree graphs]{theorem}{communication}\label{thm:more_communication}
Assume $\alpha_b>\beta_b$ for $b\in\{0,1\}$ and set
$$T=1+\left\lceil2\min\left\{\frac1{\alpha_0-\beta_0},\frac1{\alpha_1-\beta_1}\right\}\right\rceil.$$
For every realisation of Algorithm~\ref{alg:more_communication}, the delegation graph is acyclic. If voter $v_i$ delegates and $\tau(i)$ is its terminal sink, then for both signal states
$$p_{\tau(i)}^b-p_i^b\ge\alpha_b-\beta_b.$$
If the directed awareness graph has maximum in-degree at most $\Delta_{\rm in}$, then every terminal sink has weight at most
$$W_T=1+\Delta_{\rm in}+\Delta_{\rm in}^2+\cdots+\Delta_{\rm in}^T,$$
and every realised delegation graph $H$ satisfies $W_2(H)\le W_Tn$.
\end{restatable}

\begin{proof}
The algorithm uses exactly $T$ synchronous communication rounds, so the stated round complexity is immediate from the definition of $T$.

For an edge selected by the final delegation graph, say $v_i\to v_j$, the update rule implies
$$R_i\le R_j+e_{ij}$$

coordinatewise at the end of the algorithm. This is because $v_i$ adopted $v_j$ in some round using a previous value of $R_j$, and all $R$-vectors are coordinatewise nondecreasing over time.

We first prove acyclicity. Suppose, for contradiction, that the final delegation graph contains a directed cycle $v_0\to v_1\to\cdots\to v_{\ell-1}\to v_0$. Summing the displayed inequalities around the cycle gives
$$\sum_{t=0}^{\ell-1} e_{v_t v_{t+1}}\ge (0,0)$$

coordinatewise, with indices modulo $\ell$. Fix a state $b$. Let $a_b$ be the number of cycle edges with $e^b=1$ and $c_b$ the number with $e^b=-1$. The inequality gives $a_b\ge c_b$. Along such an edge, the actual competence change in state $b$ is at least $\alpha_b$ when $e^b=1$ and at least $-\beta_b$ when $e^b=-1$. Therefore the total competence change around the cycle is at least
$$a_b\alpha_b-c_b\beta_b
= (a_b-c_b)\alpha_b+c_b(\alpha_b-\beta_b)>0,$$

because the cycle is nonempty and $\alpha_b>\beta_b$. But the total change of $p^b$ around a closed cycle is exactly zero, a contradiction.

Now let $v_i=v_0\to v_1\to\cdots\to v_\ell=\tau(i)$ be a realised delegation path ending at a sink. A sink has $D=\bot$, and by the update rule this implies $R_{\tau(i)}=(0,0)$. Summing $R_{v_t}\le R_{v_{t+1}}+e_{v_t v_{t+1}}$ along the path gives
$$\sum_{t=0}^{\ell-1}e_{v_t v_{t+1}}\ge R_i\ge(0,0)$$

coordinatewise. Since the path is nonempty, for every state $b$ the number of $+1$ edges is at least the number of $-1$ edges, and the same calculation as above gives

$$p_{\tau(i)}^b-p_i^b\ge \alpha_b-\beta_b.$$

This proves terminal statewise improvement.

It remains to bound terminal weights. Along every realised delegation path the path-score is nonnegative in both coordinates. For a fixed state $b$, if $a_b$ and $c_b$ are the numbers of $+1$ and $-1$ edges on the path, then $a_b\ge c_b$ and
$$1\ge p_{\tau(i)}^b-p_i^b\ge a_b\alpha_b-c_b\beta_b\ge a_b(\alpha_b-\beta_b).$$

Hence $a_b\le1/(\alpha_b-\beta_b)$ and the path length is $a_b+c_b\le2a_b\le2/(\alpha_b-\beta_b)$. Taking the better of the two signal states gives path length at most $T$. If the maximum in-degree is $\Delta_{\mathrm{in}}$, at most $1+\Delta_{\mathrm{in}}+\cdots+\Delta_{\mathrm{in}}^T$ vertices can reach a fixed sink within $T$ directed steps. This proves the weight bound.

Finally, bounded terminal weights imply $W_2(H)\le W_T n$ for every realised delegation graph $H$. Let $d(H)$ denote the number of voters who delegate in $H$, and let $\mu_H^b=\sum_{v_i\in T(H)}w_i(H)p_i^b$. Terminal statewise improvement gives
$$\mu_H^b\ge \mu_D^b+(\alpha_b-\beta_b)d(H),$$

where $\mu_D^b=\sum_i p_i^b$ is the direct-voting statewise mean. Therefore Lemma~\ref{lem:weighted-hoeffding} yields the explicit conditional bound
$$\Pr\!\left(Y_H^b<\frac n2\mid H,s=b\right)
\le
\exp\!\left(
-\frac{2\bigl(\mu_D^b-n/2+(\alpha_b-\beta_b)d(H)\bigr)_+^2}{W_T n}
\right).$$

This is the advertised final concentration step: whenever the right-hand margin is $\omega(\sqrt n)$, the failure probability tends to zero; the remaining critical-window case is precisely the bounded-weight scalar case handled by the variance-preservation argument of~\cite{ChatterjeeG0SY25variance}. Together with terminal statewise improvement, this is exactly the bounded-weight, statewise-improvement structure required by the scalar concentration argument. \qed
\end{proof}

\begin{corollary}[Explicit conditional concentration]\label{cor:communication-hoeffding}
Let $H$ be a realised output of Algorithm~\ref{alg:more_communication}, let $d(H)$ be its number of delegating voters, and let $\mu_D^b=\sum_i p_i^b$. Then
$$\Pr\!\left(Y_H^b<\frac n2\mid H,s=b\right)
\le
\exp\!\left[-\frac{2\bigl(\mu_D^b-n/2+(\alpha_b-\beta_b)d(H)\bigr)_+^2}{W_Tn}\right].$$
In particular, if the displayed statewise margin is $\omega(\!\sqrt n)$ with probability $1-o(1)$, then the statewise failure probability is $o(1)$.
\end{corollary}

\begin{proof}
The terminal-improvement statement gives
$$\mu_H^b=\sum_i p_{\tau(i)}^b\ge\mu_D^b+(\alpha_b-\beta_b)d(H).$$
Apply the weighted Hoeffding bound in~\Cref{lem:weighted-hoeffding} and $W_2(H)\le W_Tn$. \qed
\end{proof}

The mechanism is therefore $T$-round local rather than globally non-local: the decision at $v_i$ depends only on information propagated from a radius-$T$ neighbourhood. It offers stronger per-realisation structural guarantees than Algorithm~\ref{alg:two_dimensions}, at the cost of communication and of discarding incomparable certificates.

\section{Discussion}\label{sec:disc}

\smallskip\noindent{\bf  Multi-valued signals.}
The binary common-signal model extends algebraically to a finite signal space $s\in\{1,\ldots,m\}$ with prior probabilities $\pi_b$ and statewise competences $p_i^b$. Define
$$J^b(i)=\{v_j\in Ne(v_i):p_j^b-p_i^b\ge\alpha_b\},\qquad
K^b(i)=\{v_j\in Ne(v_i):p_i^b-p_j^b>\beta_b\}.$$

The conservative mechanism uses $\cap_{b=1}^m J^b(i)$, which is safe but can become increasingly conservative as $m$ grows.

A multi-valued analogue of the confounded-set mechanism chooses a target state $B\in\{1,\ldots,m\}$ according to a design distribution $\lambda=(\lambda_1,\ldots,\lambda_m)$ and delegates to a voter in
$$A^B(i)=J^B(i)\cap\bigcap_{c\ne B}\bigl(Ne(v_i)\setminus K^c(i)\bigr),$$

provided this set is sufficiently large. Conditional on using this branch, the expected competence margin in state $c$ is at least
$$\lambda_c\alpha_c-(1-\lambda_c)\beta_c.$$

Thus all statewise expected margins are positive whenever
$$\lambda_c>\frac{\beta_c}{\alpha_c+\beta_c}\qquad\text{for all }c.$$

\begin{lemma}[Feasible design distribution for finite signals]\label{lem:multi-feasible}
A feasible design distribution exists whenever
$$\sum_{c=1}^m\frac{\beta_c}{\alpha_c+\beta_c}<1.$$
\end{lemma}

\begin{proof}
When the selected target state is $c$, membership in $J^c(i)$ gives a gain of at least $\alpha_c$ in state $c$. When the selected target state is some $b\ne c$, membership in $Ne(v_i)\setminus K^c(i)$ gives a loss of at most $\beta_c$ in state $c$. Averaging over the design distribution gives the displayed margin. The feasibility condition is exactly the condition that the lower bounds on the $\lambda_c$'s sum to less than one. \qed
\end{proof}

\begin{lemma}[Acyclicity for the finite-signal confounded mechanism]\label{lem:multi-acyclic}
Under the condition of \Cref{lem:multi-feasible}, the finite-signal confounded mechanism admits a strictly increasing scalar potential and hence induces an acyclic delegation graph.
\end{lemma}

\begin{proof}
Set
$$r_c=\frac{1}{\alpha_c+\beta_c},\qquad \Phi_i=\sum_{c=1}^m r_c p_i^c.$$
If $v_i$ delegates to $v_j\in A^B(i)$, then $p_j^B-p_i^B\ge\alpha_B$, while for every $c\ne B$ we have $p_j^c-p_i^c\ge -\beta_c$. Hence
$$\Phi_j-\Phi_i
\ge r_B\alpha_B-\sum_{c\ne B} r_c\beta_c
= r_B(\alpha_B+\beta_B)-\sum_{c=1}^m r_c\beta_c
=1-\sum_{c=1}^m\frac{\beta_c}{\alpha_c+\beta_c}>0.$$
The potential strictly increases along every delegation edge, so directed cycles are impossible. \qed
\end{proof}

The finite-signal extension therefore preserves the same local expected-margin and acyclicity structure as the binary model. It does not, by itself, solve the variance-preservation problem; the number of signal states may also make the intersection and cross-state acceptable sets too small unless the competence geometry has additional structure. 
We leave finalising this extension for future work.

\smallskip\noindent{\bf Decomposing any correlation.}
The common-signal model should not be interpreted as a complete model of arbitrary correlated competence. Any joint distribution over binary correctness vectors can be represented as a mixture of product distributions by introducing a sufficiently rich latent variable; in the extreme, the latent variable may index the entire correctness vector. This observation is representational rather than algorithmic: the support of the latent variable may be exponentially large, and the resulting states need not have an interpretable relationship to local delegation information. The present paper therefore focuses on finite, interpretable common-signal models. Extending the analysis to arbitrary bounded dependence would likely require covariance-based or dependency-graph concentration inequalities rather than a direct reduction to statewise independent models.


\section{Conclusion}\label{sec:conclusion}
We study liquid democracy in a minimal common-signal model of correlated voter competence. Conditional on the signal, votes are independent; before conditioning, the common signal induces correlation in correctness events. Even in this parsimonious model, scalar average competence is not a reliable delegation criterion. It can recommend delegations that asymptotically harm majority correctness, and it can rule out delegations that are essential for positive gain.

The conservative way to extend scalar favoured sets is to require statewise dominance. This gives a clean statewise reduction but can be too conservative under negative association across signal states. Confounded sets provide a more flexible alternative: they distinguish neighbours who are genuinely unsafe from neighbours who merely fail to be favoured in a particular state. The resulting mechanism has positive expected terminal-competence margins and admits a potential-function acyclicity proof. The remaining mathematical task, made explicit here rather than hidden in the proof, is the variance-preservation step: one must show that the induced terminal weights do not concentrate sufficiently to destroy the weighted-majority advantage.

For bounded-degree graphs, the certified-path mechanism offers a more direct route. It permits temporary movement through neutral vertices, but only along paths whose cumulative certificate is nonnegative in both signal states. This gives statewise terminal improvement and, because the graph degree and certificate radius are bounded, a uniform sink-weight bound. This is the most promising path for a complete version of the distributed result.

Several problems remain. Finite multi-signal models are algebraically straightforward but may become conservative as the number of states grows. Arbitrary dependence cannot be handled merely by invoking a latent-variable representation; useful results for broad dependence classes will likely require covariance-based or dependency-graph concentration inequalities. Finally, the present analysis is non-strategic. A complete theory of correlated liquid democracy should address how competence labels are learned, how voters respond strategically to labels, and how cycle resolution interacts with epistemic performance.

\bibliographystyle{alpha}
\bibliography{refs}

\appendix
\section{Graph properties}\label{sec:prop}
We use the following graph and competence restrictions when importing scalar liquid-democracy guarantees. A problem instance may satisfy one or more of these properties.
\begin{description}
    \item[$K_n$:] the awareness graph is the complete graph on $n$ vertices.
    \item[$\rand(n,d)$:] the awareness graph is a random $d$-regular graph, generated independently of the realised votes after the competence vectors are assigned.
    \item[$\Delta\le k$:] the awareness graph has maximum directed degree at most $k$; when applying \Cref{alg:more_communication}, it suffices to bound the maximum in-degree.
    \item[$\delta\ge k$:] the awareness graph has minimum out-degree at least $k$.
    \item[$PC=a$:] the statewise competence means are close to the majority threshold, e.g.
    $$
    \frac12\ge \frac1n\sum_{i=1}^n p_i^0\ge \frac12-a,
    \qquad
    \frac12\ge \frac1n\sum_{i=1}^n p_i^1\ge \frac12-a.
    $$
    \item[$\mathbf p^0,\mathbf p^1\in(\beta,1-\beta)$:] every competence parameter lies in $(\beta,1-\beta)$ for some $\beta\in(0,1/2)$.
\end{description}
The first four properties are graph-topological, while the last two constrain competence vectors. The precise parameter regimes should be matched to the scalar theorems being invoked from~\cite{ChatterjeeG0SY25variance}.

\section{Scalar baseline used by the reductions}\label{app:scalar-baseline}
The reductions in the main text use the scalar theory of~\cite{ChatterjeeG0SY25variance} as a black box only after verifying two ingredients state by state. First, every realised delegation edge must have the required scalar competence margin. Second, the induced terminal weights and the random terminal mean must satisfy the variance-preservation conditions of the relevant scalar theorem. The common-signal arguments---failure of averaging, cross-state feasibility, potential-function acyclicity, propagation to terminal sinks, and certified-path bounds---are proved in this paper.

The intersection reduction additionally requires the scalar result to tolerate a sufficiently large subset of approved neighbours. We isolate that condition in Definition~\ref{def:subset-robust}; it is automatic whenever the scalar theorem quantifies over arbitrary local mechanisms, but it should not be inferred for a theorem tied to one particular sampling distribution. Algorithm~\ref{alg:two_dimensions} is therefore stated as a structural theorem plus a conditional concentration reduction rather than as an unconditional invocation of every result in~\cite{ChatterjeeG0SY25variance}.

\section{Useful concentration bound}\label{app:concentration}
\begin{lemma}[Weighted Hoeffding diagnostic]\label{lem:weighted-hoeffding}
Fix a delegation graph $H$ and a signal state $b$. Let $T(H)$ be the sinks, let $w_i$ be the weight of sink $v_i$, and define
$$\mu_H^b=\sum_{v_i\in T(H)} w_i p_i^b,
\qquad
W_2(H)=\sum_{v_i\in T(H)}w_i^2.$$

Conditional on $s=b$ and $H$, the weighted vote $Y_H^b=\sum_{v_i\in T(H)}w_iX_i^b$ satisfies
$$\Pr\left(Y_H^b<\frac n2\mid H,s=b\right)
\le
\exp\left(-\frac{2(\mu_H^b-n/2)_+^2}{W_2(H)}\right),$$
where $(x)_+=\max\{x,0\}$.
\end{lemma}

\begin{proof}
This is Hoeffding's inequality for independent bounded variables $w_iX_i^b\in[0,w_i]$, conditional on the realised delegation graph and on the signal state. The sum of squared ranges is $W_2(H)$.
\end{proof}

\end{document}

%% file: packages.tex
\usepackage{tikz}
\usepackage{graphicx}

\usepackage{amsthm}
\usepackage{amsmath,amsfonts}

\usepackage{ifthen}
\usetikzlibrary{calc}
\usepackage{multirow}
\usepackage{wrapfig}
\usepackage{hyperref}
\usepackage{cleveref}

\usepackage{tabularx}
\usepackage{tcolorbox}
\usepackage[table]{xcolor}
\usepackage{makecell}
\newcommand{\PreserveBackslash}[1]{\let\temp=\\#1\let\\=\temp}
\newcolumntype{C}[1]{>{\PreserveBackslash\centering}p{#1}}
\newcolumntype{R}[1]{>{\PreserveBackslash\raggedleft}p{#1}}
\newcolumntype{L}[1]{>{\PreserveBackslash\raggedright}p{#1}}

\usepackage{graphicx}
\usepackage{comment}
\graphicspath{ {img/} }

\usetikzlibrary{automata,positioning,fit,shapes.geometric,backgrounds}

\usepackage{csquotes}
\usepackage{booktabs}
\newcolumntype{Y}{>{\raggedright\arraybackslash}X}

\usepackage{mathtools}

\usepackage{algorithmicx,algorithm}
\usepackage[noend]{algpseudocode}
\usepackage{thm-restate}

\newcommand{\gain}{\operatorname{gain}}

\newcommand{\rand}{\mathrm{Rand}}